\documentclass[nofootinbib,prb,amsmath,amssymb,twocolumn,floatfix]{revtex4-2}

\usepackage{graphicx}
\usepackage{dcolumn}
\usepackage{bm}
\usepackage{siunitx}

\usepackage{amsmath, amsfonts, amsthm}
\usepackage{subcaption}
\newtheorem{theorem}{Theorem}
\newtheorem{lemma}[theorem]{Lemma}

\theoremstyle{definition}

\newcommand{\de}{\,\mathrm{d}}

\newcommand{\dd}[2]{\frac{\mathrm{d}#1}{\mathrm{d}#2}}
\newcommand{\Znn}{\mathbb{Z}^{\geq0}}
\newcommand{\Zp}{\mathbb{Z}^{>0}}
\newcommand{\Hill}{\mathcal{H}}
\renewcommand{\i}{\mathrm{i}}
\newcommand{\F}{\mathcal{F}}
\newcommand{\tr}{\mathrm{tr}}

\begin{document}

\preprint{APS/123-QED}

\title{Van Hove singularities in the density of states of a chaotic dynamical system}

\author{Bryn Davies}
\affiliation{Mathematics Institute, University of Warwick, Coventry CV4~7AL, United Kingdom }

\date{\today}

\begin{abstract}
	We show that the statistics of chaotic systems can be predicted by constructing an associated sequence of periodic differential operators and computing their densities of states. For such operators, the density of states is well understood and can be computed straightforwardly, often yielding explicit formulas. As a case study, we investigate a nonlinear recursion relation that maps naturally onto a family of periodic operators generated by a Fibonacci tiling rule. This correspondence enables us to derive an explicit formula for the limiting statistics of the chaotic system and to demonstrate that the clustering near to critical values is equivalent to the van Hove singularities in the operators' densities of states.
\end{abstract}

\maketitle

\section{Introduction}

Chaotic systems have been long-standing objects of fascination. Although the  long-term behaviour of individual orbits is hard to predict, in no small part due to the sensitivity to initial conditions, the collective behaviour of many orbits can often be captured in a statistical sense by understanding convergence to invariant measures.  Then, using probabilistic methods, it is often possible to make precise statements about the likelihood of a future state \cite{lasota2013chaos, walters2000introduction}. Although this is most commonly achieved by computing the density of states through exploiting ergodic properties and finding invariant measures, other methods exist, such as understanding a system's statistics by reformulating the problem in terms of random matrix theory \cite{erdHos2017dynamical, haake1991quantum}. 



In this work, we demonstrate a new approach to predicting the statistics of a chaotic system, by constructing a corresponding sequence of periodic differential operators and computing their associated densities of states. This ambition is similar in spirit to moment problems, where the aim is to find a linear operator whose moments generate the values of the recursion relation \cite{akhiezer2020classical, el2004recursive, taher2009complex}. However, in this case our approach will have the significant benefit of being able to exploit the well-developed spectral theory for periodic differential operators \cite{kuchment2016overview}. This yields predictions of the system's statistics including, as in the case study considered here, explicit formulas.

The link between dynamical systems and spectral properties of one-dimensional differential operators is well known. It has facilitated many breakthroughs on the spectral theory of operators with exotic and non-periodic coefficients. For example, much of the spectral theory for quasi-periodic operators is based on reformulating problems as dynamical systems \cite{simon1982almost, avila2009ten, surace1990schrodinger} and random systems are often studied using Lyapunov functions to characterise decay rates and mean free paths \cite{borland1963nature, carmona1982exponential, comtet2013lyapunov, scales1997lyapunov}. In this work, we will show that this relationship is also useful in the other direction, as reformulating a dynamical system as a sequence of differential operators allows well-known formulas for their densities of states to be used to predict the statistics.

We will take as a case study a nonlinear recursion relation which can be related to a sequence of periodic operators generated by a Fibonacci tiling rule. The recursion relation is 
\begin{equation} \label{eq:recursion}
	x_{n+1} = x_n x_{n-1} - x_{n-2},
\end{equation}
where the initial conditions $x_0$, $x_1$ and $x_2$ are taken to be real numbers. The setting considered in this work is where $x_0$ and $x_1$ are independent random variables and $x_2$ is chosen to depend on $x_0$ and $x_1$. The choice of this dependence is crucial to our method, as we need the relationship to be such that we are able to form an associated sequence of differential operators. In this work, we will focus on relationships that correspond to multiplication of monodromy matrices, as this will generate \eqref{eq:recursion} via a Fibonacci tiling rule. Under these specific, carefully chosen choices of $x_2$, the system converges to a limiting distribution with strong clustering around the critical values $x=\pm2$ (shown on the right in Fig.~\ref{fig:hero}). We will see that this clustering can be linked to van Hove singularities in the densities of states of the associated sequence of differential operators. 

\begin{figure}
	\includegraphics[width=\linewidth,trim=0 0 0 0.5cm,clip]{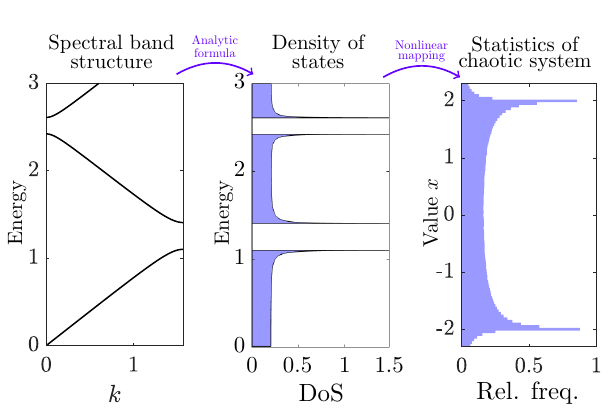}
	\caption{We predict the statistics of a chaotic system by constructing a map from the spectra of a sequence of periodic differential operators. Such operators have spectra composed of continuous bands and analytic formulas for the densities of states (DoS).} \label{fig:hero}
\end{figure}

\begin{figure*}
	\begin{subfigure}[b]{0.3\textwidth}
		\includegraphics[width=0.95\linewidth]{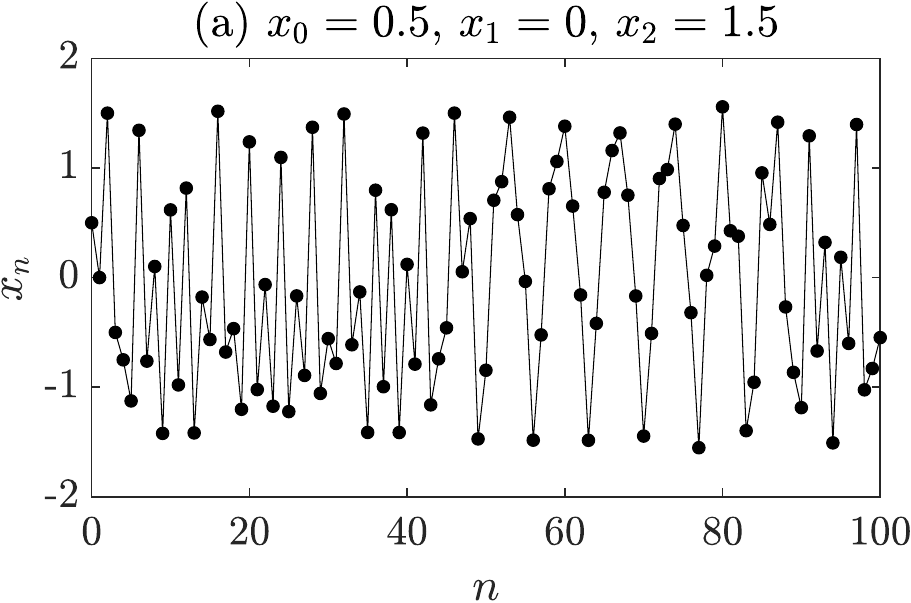}
	\end{subfigure}
	\begin{subfigure}[b]{0.3\textwidth}
		\includegraphics[width=0.95\linewidth]{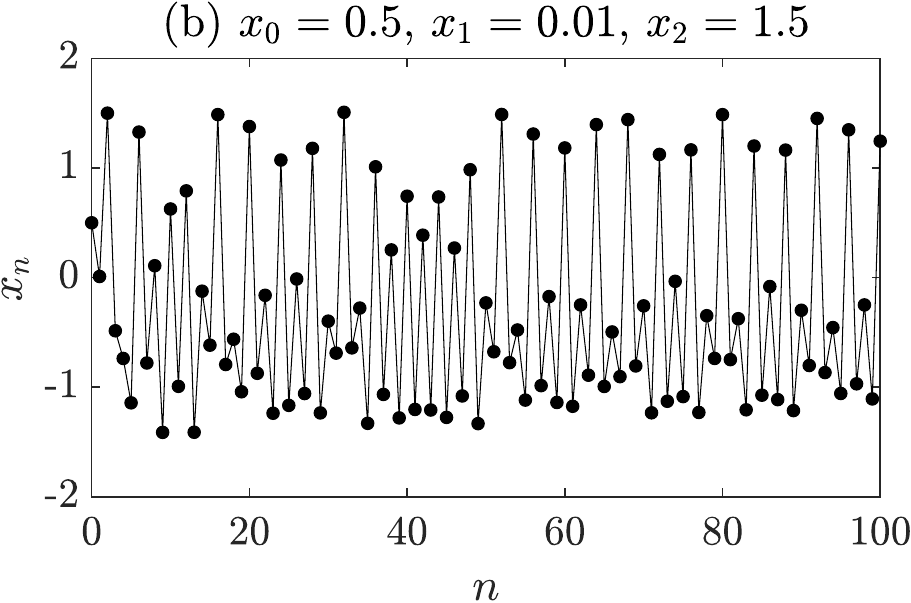}
	\end{subfigure}
	\begin{subfigure}[b]{0.3\textwidth}
		\includegraphics[width=0.95\linewidth]{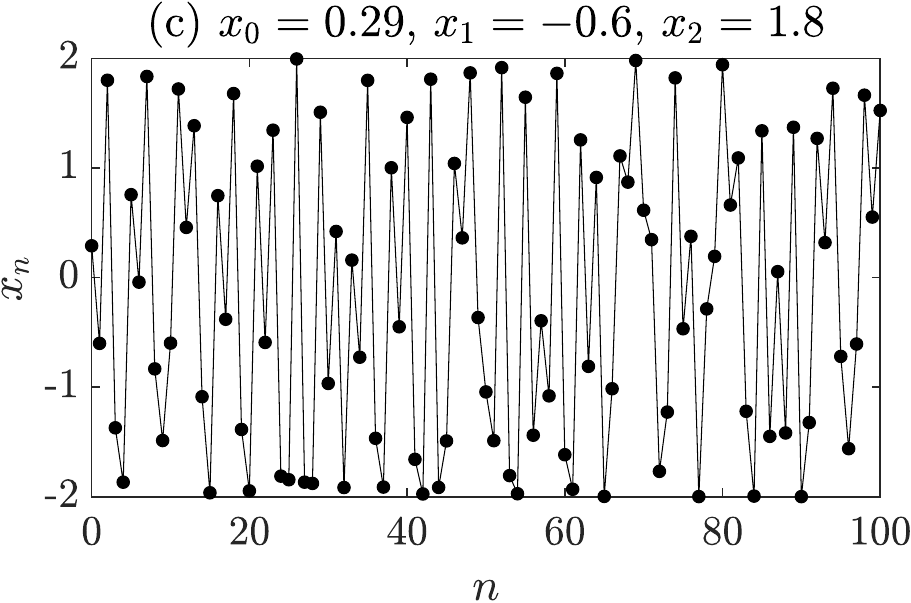}
	\end{subfigure}
	
	\vspace{0.3cm}
	
	\begin{subfigure}[b]{0.3\textwidth}
		\includegraphics[width=0.95\linewidth]{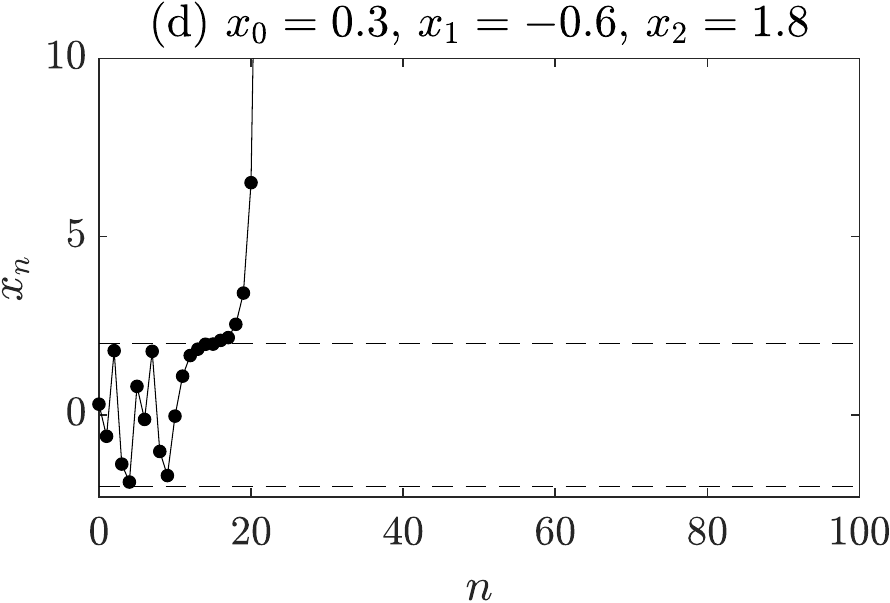}
	\end{subfigure}
	\begin{subfigure}[b]{0.3\textwidth}
		\centering
		\includegraphics[width=0.95\linewidth]{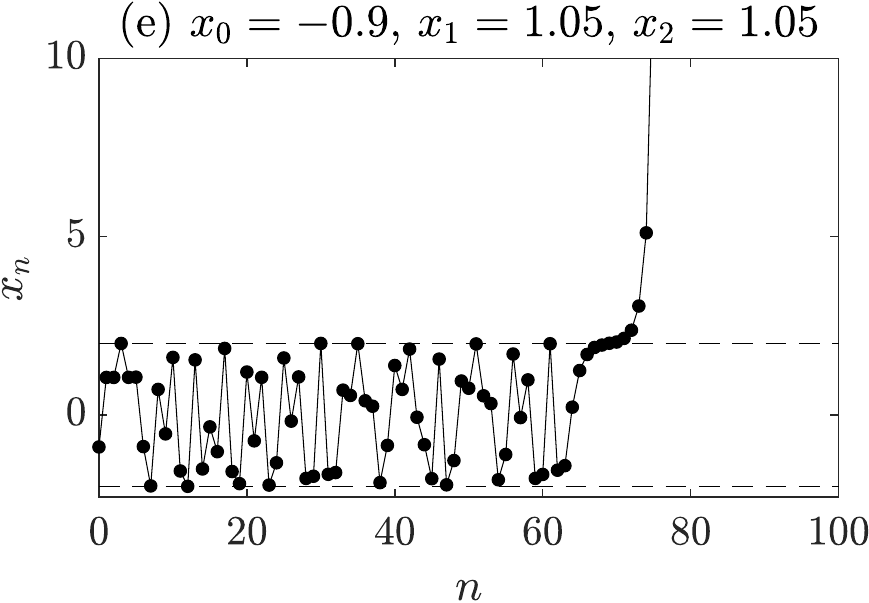}
	\end{subfigure}
	\begin{subfigure}[b]{0.3\textwidth}
		\includegraphics[width=0.95\linewidth]{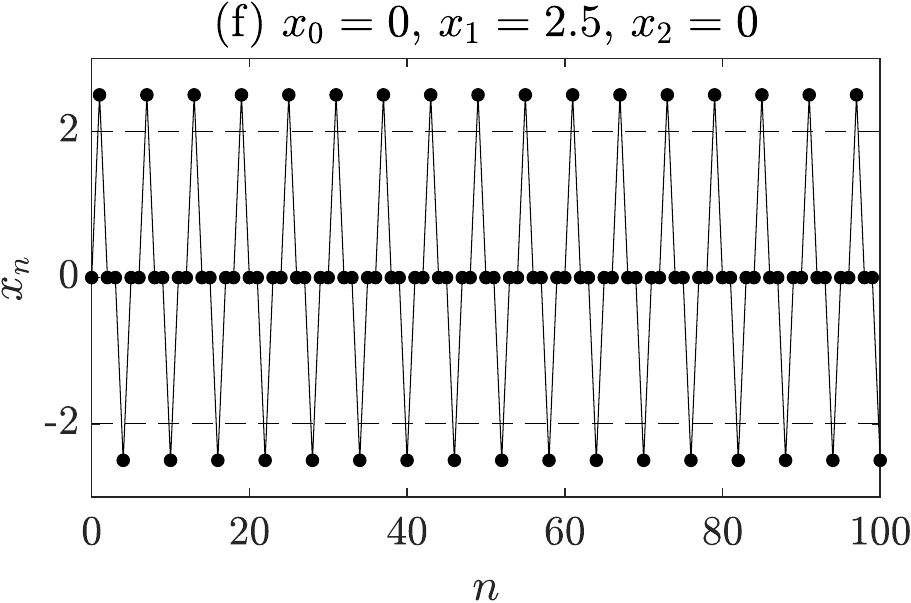}
	\end{subfigure}
	\caption{Examplar orbits of the recursion relation (\ref{eq:recursion}) for different initial conditions. Dashed lines indicate the critical threshold $|x_n|=2$. Examples (a)--(c) are bounded by this threshold (for the 100 terms shown), while (d) and (e) diverge after exceeding it. (f) shows an example of a periodic orbit that exceeds this threshold every third term.} \label{fig:examples}
\end{figure*}

Fibonacci-generated patterns have been studied extensively in a variety of settings. They are a canonical example of a tiling rule that generates a quasicrystal \cite{kolavr1993new, levine1984quasicrystals}. They also tend to lead to operators with exotic spectral properties \cite{kohmoto1984cantor, davies2023super, sutHo1989singular,sutHo1989singular, davies2025convergence, kohmoto1983localization, morini2018waves}, which has inspired their use in designing a variety of quantum and classical wave systems  \cite{jagannathan2021fibonacci, verbin2015topological, davies2022symmetry, gei2020phononic, hamilton2021effective, morini2019negative}. Such materials have been one of the foremost examples in the development of quasi-periodic metamaterials for wave control applications \cite[Section~11]{davies2025roadmap}.

Although the sequence of periodic operators generated by the Fibonacci tiling rule presents interesting and challenging questions about how their spectra evolve and are related to one another \cite{kohmoto1984cantor, davies2023super, sutHo1989singular,sutHo1989singular, davies2025convergence, kohmoto1983localization, morini2018waves}, any individual operator in the sequence has a Bloch spectrum that can be understood relatively easily. Not only can its spectral bands be computed straightforwardly using standard Bloch methods, but the density of states within these bands can be computed using an analytic formula (as the inverse of the slope of the spectral bands) \cite{kuchment2016overview}. By understanding how these states are mapped to values of the dynamical system, its statistics can be predicted. This idea is sketched in Fig.~\ref{fig:hero}. As well as producing an analytic formula for the statistics, this method also reveals how the singularities shown in the statistics in Fig.~\ref{fig:hero} are equivalent to van Hove singularities in the operators' densities of states \cite{van1953occurrence}. These are singularities that occur at the critical points in the spectra of periodic operators, such as at the edges of the spectral bands of one-dimensional operators \cite{reed1978iv}.

This paper will begin by recalling the escape criteria that characterise when an orbit of \eqref{eq:recursion} diverges. Since these criteria are well studied in the literature, we will focus on characterising the statistics of the bounded orbits. In Section~\ref{sec:periodic} we construct an associated sequence of periodic operators using the Fibonacci tiling rule, which we use in Section~\ref{sec:density} to derive our main result, which is an explicit formula for the distribution of bounded orbits.

\section{Escape criteria} \label{sec:basicprops}

We will begin by summarising some of the well-known properties of the recursion relation \eqref{eq:recursion}. Some exemplar orbits are shown in Fig.~\ref{fig:examples}. In each case, the initial values $\{x_i\}_{i=0,1,2}$ are chosen to lie within the interval $[-2,2]$. It is apparent that the behaviour typically takes one of three forms. In some cases, the orbits appear to stay bounded, fluctuating erratically within the interval $[-2,2]$, as shown in Figs.~\ref{fig:examples}(a)--(c). However, for some initial conditions the orbit escapes the interval $[-2,2]$ and then diverges quickly, as shown in Figs.~\ref{fig:examples}(d)--(e). Finally, there also exist periodic orbits. One is depicted in Fig.~\ref{fig:examples}(f), which is an example of the 6-periodic orbits that will occur if two of the three initial values are zero. Notice that the repeating non-zero value can be arbitrarily large, depending on the initial values. 


Escape criteria for this system have been established in the literature. The existence of periodic orbits as in Fig.~\ref{fig:examples}(f) shows that leaving the interval $[-2,2]$ is not sufficient to guarantee divergence, however it was shown in \cite[Theorem~3.1]{davies2023super} that if we also have three successively growing terms, in the sense that
\begin{equation} \label{eq:growth}
	2<|x_N|<|x_{N+1}|<|x_{N+2}|,
\end{equation}
for some $N$, then the orbit is guaranteed to diverge\footnote{Since the version of this result proved in \cite[Theorem 3.1]{davies2023super} is slightly weaker than showing divergence, a modification of the argument is included here in Appendix~\ref{app:escape},  for completeness.}.

There has been particular interest in studying the cases where the system can be related to the spectra of periodic operators. This happens when $x_2$ is chosen to depend suitably on $x_0$ and $x_1$, as we will study below. In this case, the set of initial conditions for which the recursion relation (\ref{eq:recursion}) remains within the bounded interval $[-2,2]$ is known to be a Cantor set (\emph{i.e.} a closed set with no isolated points and whose complement is dense) \cite{kohmoto1984cantor, sutHo1989singular}. Further, this set can have zero Lebesgue measure \cite{sutHo1989singular} and display the self-similar properties typical of fractals \cite{kohmoto1983localization, morini2018waves}. Additionally, under these assumptions on the initial conditions, the escape criterion \eqref{eq:growth} can be strengthened \cite[Corollary 3.10]{davies2025convergence}.

While \eqref{eq:growth} characterises the divergence that can occur when an orbit leaves the critical interval $[-2,2]$, the question considered in this work is how states behave while they remain within this interval. We will use the link to periodic operators to characterise the statistics and understand the chaotic behaviour of the bounded orbits shown in Fig.~\ref{fig:examples}(a)--(c).

\section{Statistics of bounded orbits} \label{sec:bounded}

A natural first question to ask is how the orbits of (\ref{eq:recursion}) are distributed on the interval $[-2,2]$ when the initial conditions $\{x_i\}_{i=0,1,2}$ are independent random variables. For example, suppose they are independent uniform random variables on $[-2,2]$. In this case, histograms showing the distribution of the orbits are given in Fig.~\ref{fig:indep}. Figs.~\ref{fig:indep}(a)--(c) show the distribution of $x_3$, $x_4$ and $x_5$ restricted to $[-2,2]$, with each histogram showing the result of $10^5$ independent realisations of the initial conditions. Then, Fig.~\ref{fig:indep}(d) shows the limiting statistics of the recursion relation: for each of the $10^5$ realisations of the initial conditions, the first 200 iterations of the recursion relation are computed and all those falling within the interval $[-2,2]$ are plotted. In this case, it is apparent that the values tend to spread across the interval $[-2,2]$ with a bias towards the centre. The ``rounded'' nature of the distribution is already beginning to emerge from the distributions of the first few values, as depicted in Fig.~\ref{fig:indep}(a)--(c). The distribution of $x_3=x_2x_1-x_0$ can be calculated explicitly, using routine calculations for products and sums of random variables (see Appendix~\ref{app:pdf_indep}):
\begin{equation} \label{eq:indepdist}
	f_{X_3}(x)= \tfrac{x-2}{32}\log\left(\tfrac{2-x}{4}\right) - \tfrac{x+2}{32}\log\left(\tfrac{x+2}{4}\right) + \tfrac{1}{8},
\end{equation}
for $-2\leq x\leq2$. This is indicated on Fig.~\ref{fig:indep}(a) by the dashed line. In the limit the statistics converge to the semicircle distribution with radius 2, given by
\begin{equation} \label{dist:semicircle}
	D(x)=\frac{1}{2\pi} \sqrt{4-x^2}.
\end{equation}
This is shown by the dashed line in Fig.~\ref{fig:indep}(d). 

\begin{figure}
	\centering
	\begin{subfigure}[b]{0.485\linewidth}
		\centering
		\includegraphics[width=0.95\linewidth,trim=1.5cm 1.5cm 0 1cm,clip]{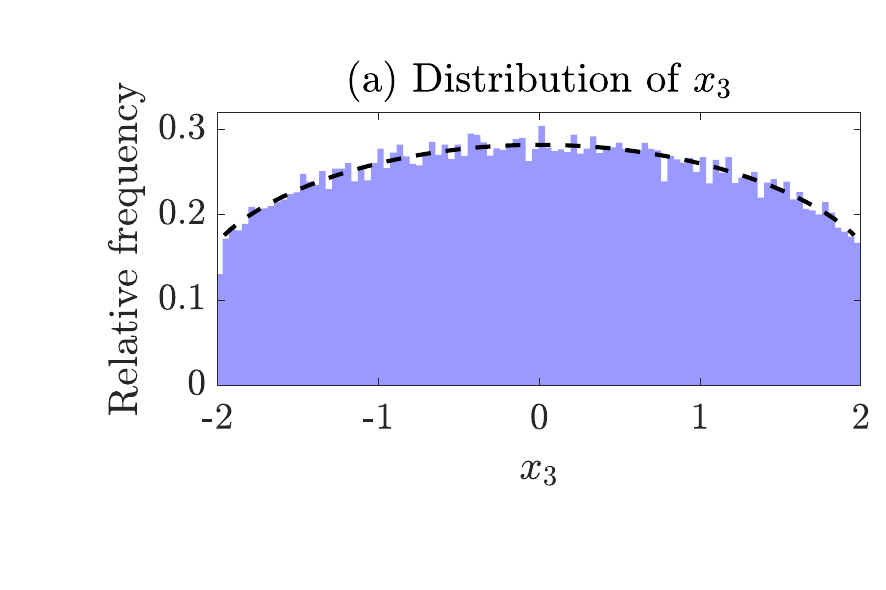}
	\end{subfigure}
	\begin{subfigure}[b]{0.485\linewidth}
		\centering
		\includegraphics[width=0.95\linewidth,trim=1.5cm 1.5cm 0 1cm,clip]{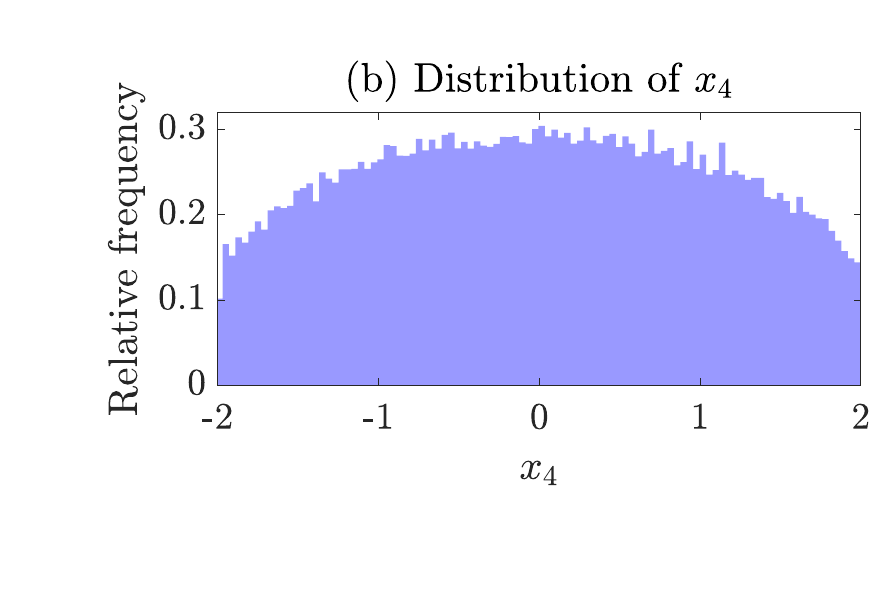}
	\end{subfigure}
	
	\vspace{0.2cm}
	
	\begin{subfigure}[b]{0.485\linewidth}
		\centering
		\includegraphics[width=0.95\linewidth,trim=1.5cm 1.5cm 0 1cm,clip]{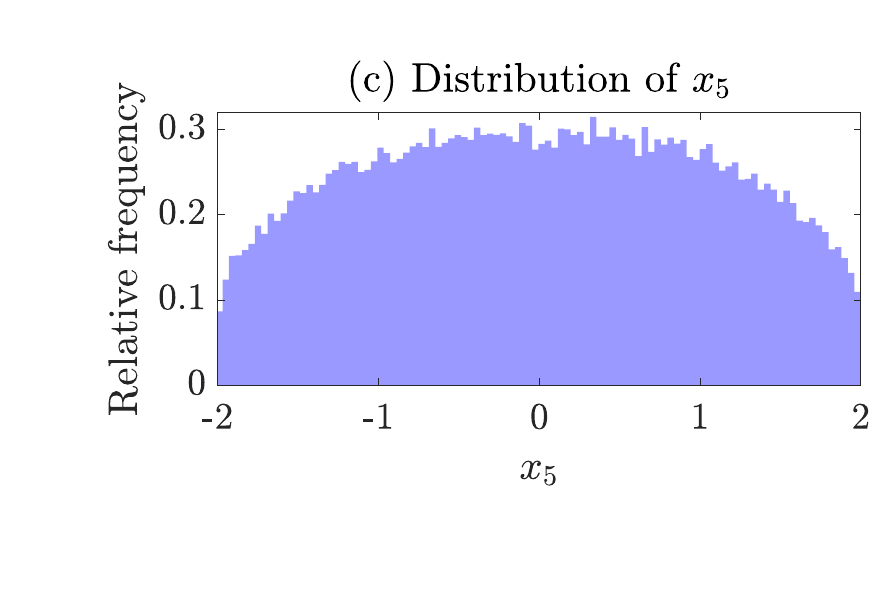}
	\end{subfigure}
	\begin{subfigure}[b]{0.485\linewidth}
		\centering
		\includegraphics[width=0.95\linewidth,trim=1.5cm 1.5cm 0 1cm,clip]{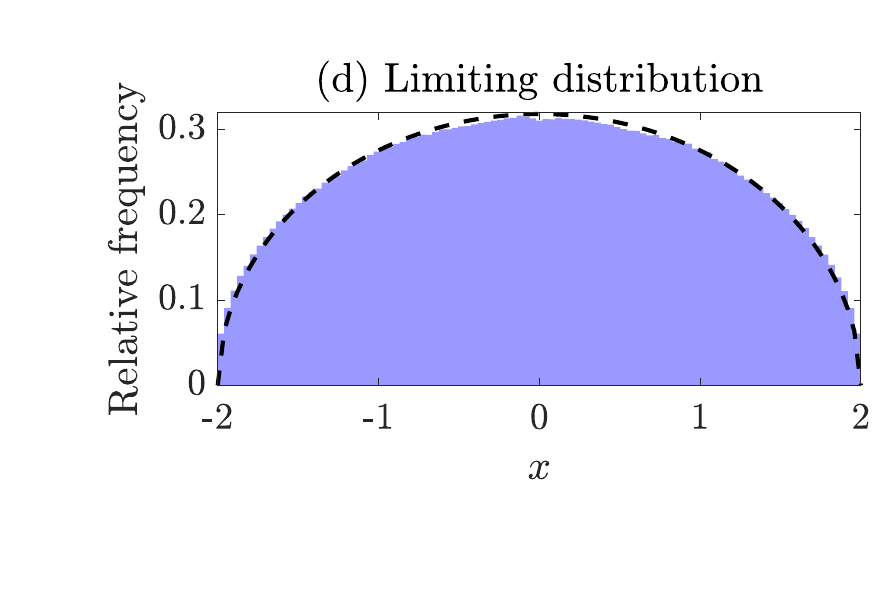}
	\end{subfigure}
	\caption{The statistics of the recursion relation (\ref{eq:recursion}) when the initial values $x_0$, $x_1$ and $x_2$ are independent and uniformly distributed on $[-2,2]$. The dashed lines in (a) and (d) show the distributions from equations \eqref{eq:indepdist} and (\ref{dist:semicircle}), respectively. The limiting distribution contains the first 200 terms. $10^5$ independent realisations of the initial conditions are used for each histogram. } \label{fig:indep}
\end{figure}

Suppose, now, that we allow $x_2$ to depend on the proceeding values $x_0$ and $x_1$ (which we will always suppose are independent and uniformly distributed on $[-2,2]$). The trend of values spreading across the interval $[-2,2]$ with a bias towards the centre persists for many choices of dependencies. For example, Figs.~\ref{fig:example_stats}(a)--(c) show the limiting statistics when $x_2$ is chosen as the arithmetic mean, geometric mean and maximum of $x_0$ and $x_1$, respectively. In each case, we see a distribution with its peak in a neighbourhood of zero.

\begin{figure}
	\centering
	\begin{subfigure}[b]{0.485\linewidth}
		\centering
		\includegraphics[width=0.95\linewidth,trim=1.5cm 1.5cm 0 1cm,clip]{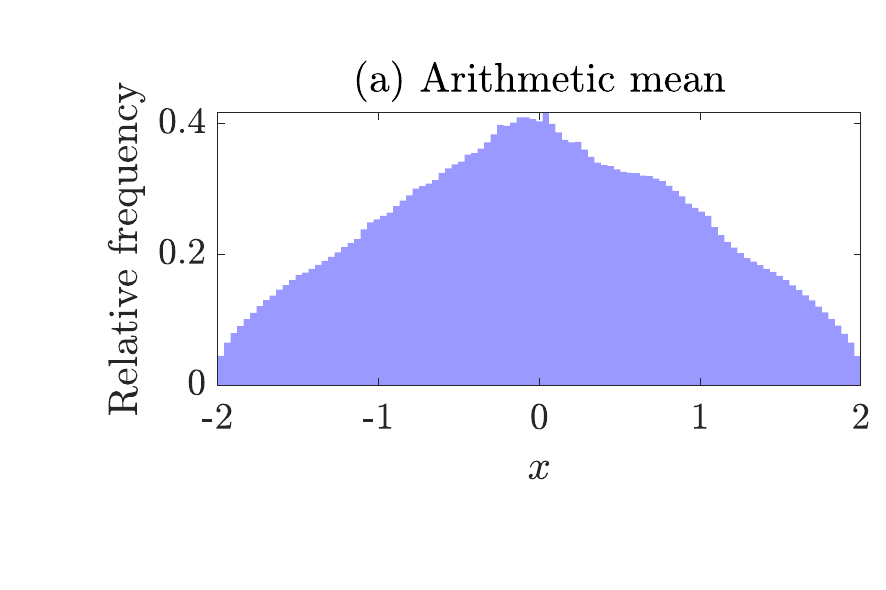}
	\end{subfigure}
	\begin{subfigure}[b]{0.485\linewidth}
		\centering
		\includegraphics[width=0.95\linewidth,trim=1.5cm 1.5cm 0 1cm,clip]{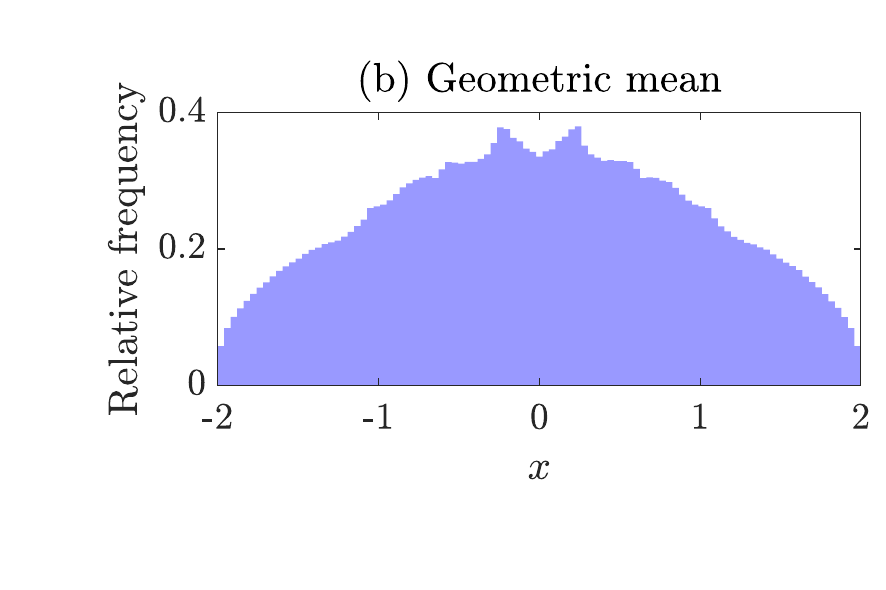}
	\end{subfigure}
	
	\vspace{0.2cm}
	
	\begin{subfigure}[b]{0.485\linewidth}
		\centering
		\includegraphics[width=0.95\linewidth,trim=1.5cm 1.5cm 0 1cm,clip]{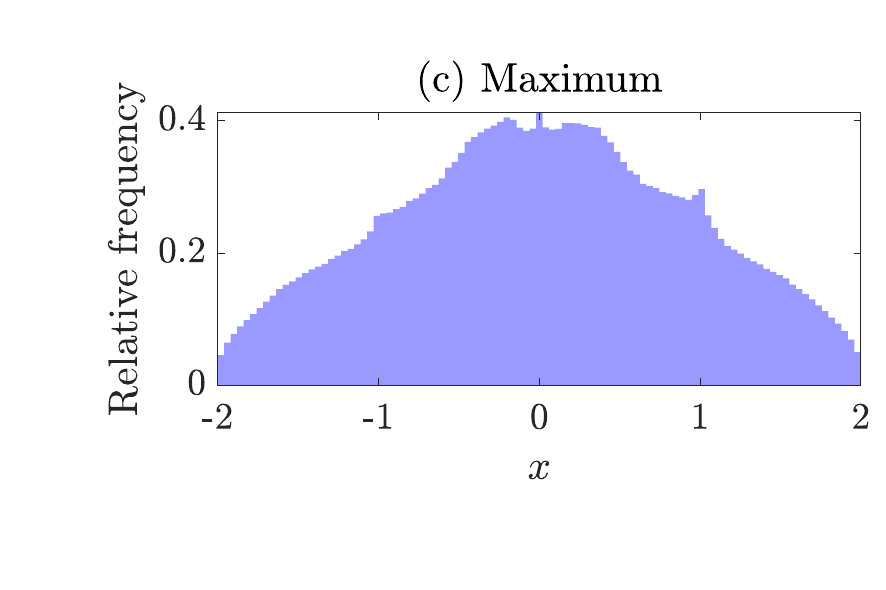}
	\end{subfigure}
	\begin{subfigure}[b]{0.485\linewidth}
		\centering
		\includegraphics[width=0.95\linewidth,trim=1.5cm 1.5cm 0 1cm,clip]{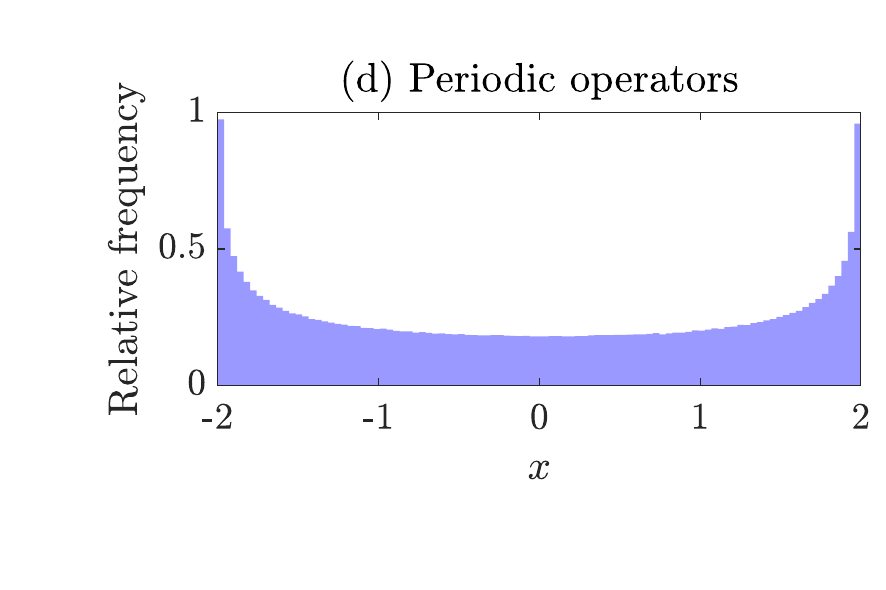}
	\end{subfigure}
	\caption{The statistics of the recursion relation (\ref{eq:recursion}) when $x_2$ is chosen to depend on $x_0$ and $x_1$, which are both independent and uniformly distributed on $[-2,2]$. In (a) $x_2=(x_0+x_1)/2$. In (b) $x_2=\mathrm{sgn}(x_0x_1)\sqrt{x_0x_1}$. In (c) $x_2=\max\{x_0,x_1\}$. In (d), the formula \eqref{eq:x2} is used. Each distribution contains the first 200 terms and $10^5$ independent realisations of the initial conditions.} \label{fig:example_stats}
\end{figure}

The behaviour of values clustering close to zero is not universal and certain choices for the dependence of $x_2$ on $x_0$ and $x_1$ lead to very different behaviour. The class of dependencies that we are interested in are choices that allow the system to be related to a sequence of periodic operators. Fig.~\ref{fig:example_stats}(d) shows an example of the statistics in this case, where we see that the behaviour is quite different as the values cluster strongly at the edges of the interval $[-2,2]$.

The class of dependencies (of $x_2$ on $x_0$ and $x_1$) that allow the system to be related to a sequence of periodic operators is broad and difficult to understand in general. To devise a suitable formula, we want to find two periodic differential operators whose monodromy matrices $M_0$ and $M_1$ (defined below) are such that $x_0=\tr(M_0)$ and $x_1=\tr(M_1)$. Then, the formula for $x_2$ is that we need $x_2=\tr(M_1M_0)$. Although this class is challenging to understand, in  general, specific examples are easy to generate. The example shown in Fig.~\ref{fig:example_stats}(d) is obtained by defining the quantities $\omega_0=\arccos(x_0/2)$ and $\omega_1=\arccos(x_1/2)$, then the formula of interest for $x_2$ is 
\begin{equation} \label{eq:x2}
	x_2 = 2\cos\omega_0\cos\omega_1-\big( \tfrac{\omega_0}{\omega_1} + \tfrac{\omega_1}{\omega_0} \big) \sin\omega_0\sin\omega_1.
\end{equation}
We will show, in Section~\ref{sec:constants}, that this corresponds to the canonical example of differential operators with piecewise constant coefficients. Another example, which uses Mathieu functions, is derived in Appendix~\ref{app:mathieu}.

In the case of the example \eqref{eq:x2}, the clustering is already partly visible in the distribution of $x_2$, shown in Fig.~\ref{fig:histograms}(a), where the values are skewed towards -2. Although this behaviour is difficult to estimate precisely, it is explored in more detail in Appendix~\ref{app:pdf_x2}. For subsequent iterations, the distribution is more symmetric, and the distribution eventually converges to a symmetric distribution with singularities at the edges of $[-2,2]$, shown in Fig.~\ref{fig:histograms}(d). To understand this, and derive the analytic estimate shown by a dotted line in Fig.~\ref{fig:histograms}(d), we will explore the underlying link with periodic differential operators.



\begin{figure}
	\centering
	\begin{subfigure}[b]{0.485\linewidth}
		\centering
		\includegraphics[width=0.95\textwidth,trim=1.5cm 1.5cm 0 1cm,clip]{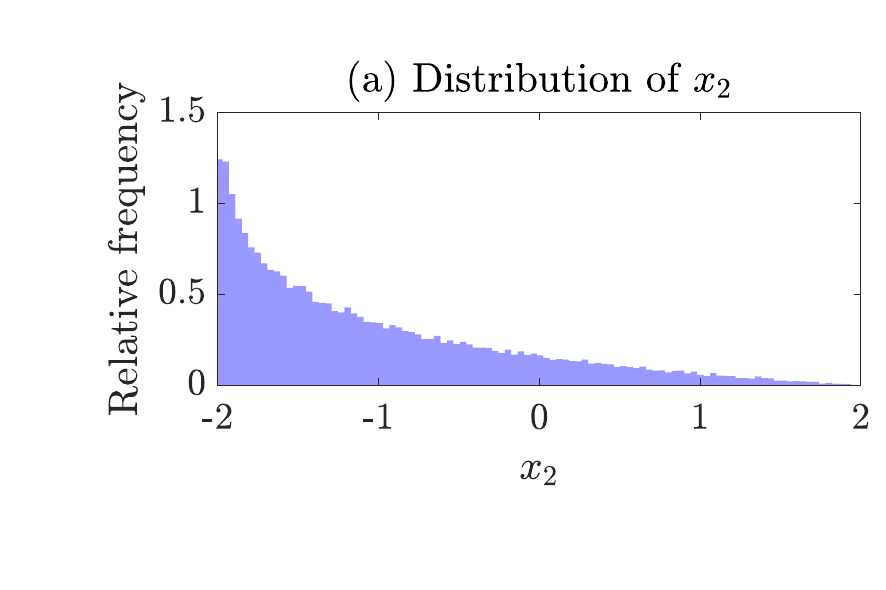}
	\end{subfigure}
	\begin{subfigure}[b]{0.485\linewidth}
		\centering
		\includegraphics[width=0.95\textwidth,trim=1.5cm 1.5cm 0 1cm,clip]{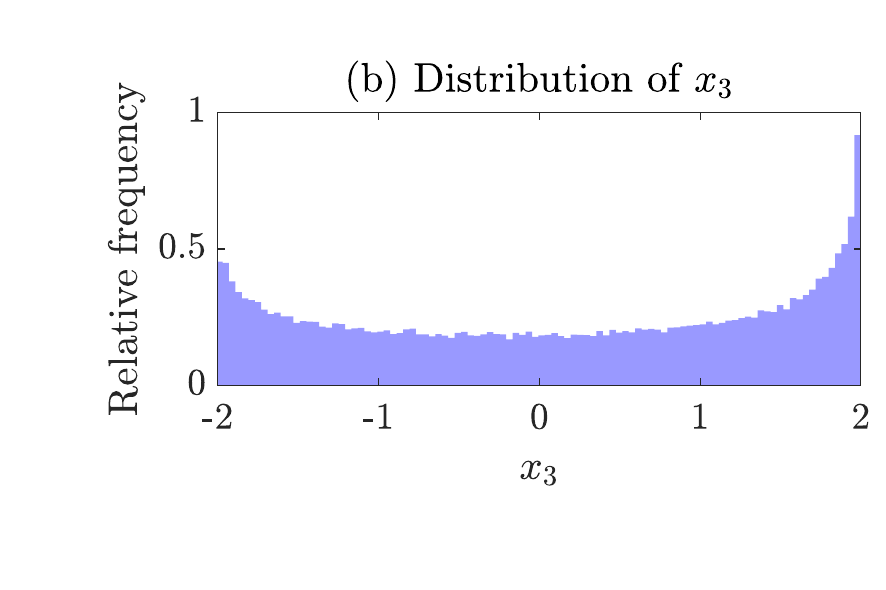}
	\end{subfigure}
	
	\vspace{0.2cm}
	
	\begin{subfigure}[b]{0.485\linewidth}
		\centering
		\includegraphics[width=0.95\textwidth,trim=1.5cm 1.5cm 0 1cm,clip]{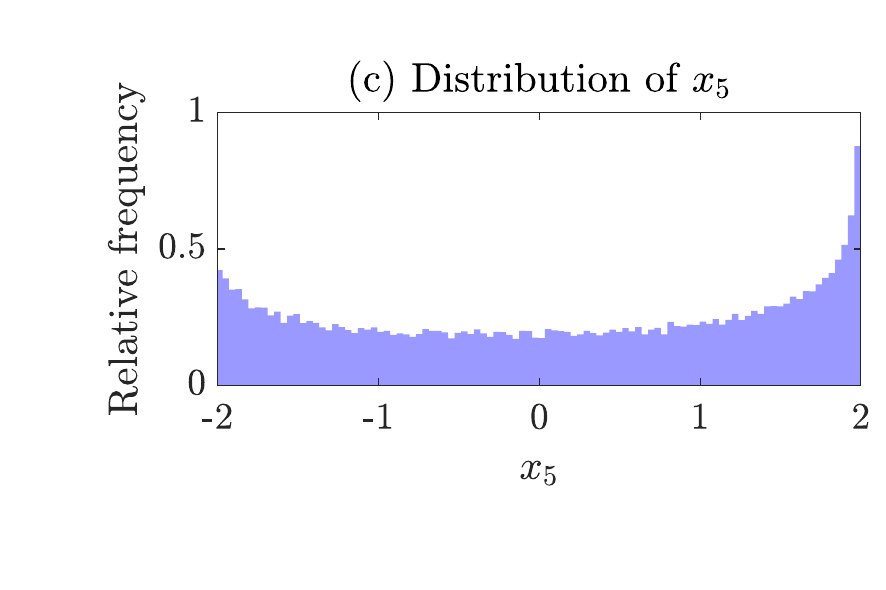}
	\end{subfigure}
	\begin{subfigure}[b]{0.485\linewidth}
		\centering
		\includegraphics[width=0.95\textwidth,trim=1.5cm 1.5cm 0 1cm,clip]{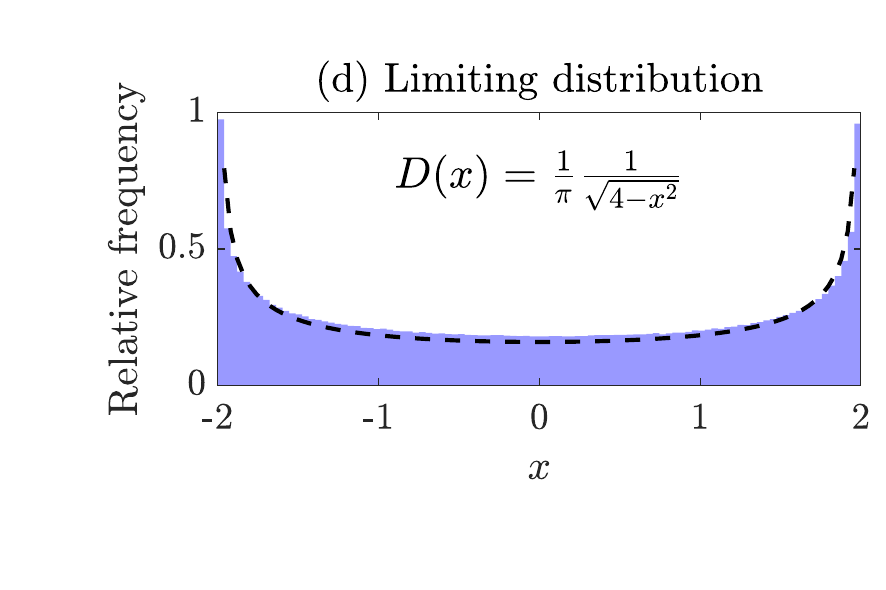}
	\end{subfigure}
	\caption{The distribution of states of the recursion relation (\ref{eq:recursion}) when the initial values $x_0$ and $x_1$ are independent and uniformly distributed on $[-2,2]$ and $x_2$ is chosen to depend on $x_0$ and $x_1$ according to (\ref{eq:x2}). The limiting statistics contain the first 200 terms and $10^5$ independent realisations of the initial conditions. The dashed line in (d) is the analytic formula \eqref{eq:main} obtained from density of states calculations.} \label{fig:histograms}
\end{figure}

\section{Periodic operators} \label{sec:periodic}

Consider the differential eigenvalue problem 
\begin{equation} \label{eq:eigenproblem}
	\Hill_L u=\lambda u, \qquad
		\Hill_L=\dd{^2}{x^2}+V(x),
\end{equation}
where  $V:\mathbb{R}\to\mathbb{R}$ is some  piecewise smooth and $L$-periodic potential function. A key tool for understanding the solution space of (\ref{eq:eigenproblem}) is the \emph{monodromy matrix}. Given the basis functions $\phi(x,\lambda)$ and $\psi(x,\lambda)$, which are defined as solutions to (\ref{eq:eigenproblem}) satisfying the conditions $\phi(0,\lambda)=\partial_x \psi(0,\lambda)=1$ and $\partial_x \phi(0,\lambda)= \psi(0,\lambda)=0$, the monodromy matrix $M_L$ is defined as
\begin{equation} \label{eq:monod_mat}
	M_L(\lambda):=\begin{pmatrix} \phi(L,\lambda) & \psi(L,\lambda) \\ \partial_x\phi(L,\lambda) & \partial_x\psi(L,\lambda) \end{pmatrix}.
\end{equation}
$M_L(\lambda)\in\mathrm{SL}(2,\mathbb{R})$ is a square matrix with determinant equal to 1.

The crucial information about the spectrum of the differential operator $\Hill_L$ is given by the eigenvalues of $M_L$, which are in turn determined by its trace. As a result, the trace of the monodromy matrix
\begin{equation} \label{eq:discriminant}
	\Delta_L(\lambda) := \mathrm{tr} \, M_L(\lambda) = \phi(L,\lambda)+\partial_x\psi(L,\lambda),
\end{equation}
will be the key quantity arising from the analysis of these operators. It is often known as the \emph{discriminant} of the operator $\Hill_L$. 
The dispersion relation of the differential eigenvalue problem \eqref{eq:eigenproblem} is 
\begin{equation} \label{eq:dispersion1}
	\det\left( M_L(\lambda) - e^{\i k L}I \right)=0,
\end{equation}
where $k$ is the Bloch momentum.  Using the fact that $\det(M_L)=1$, this is equivalent to
\begin{equation} \label{eq:dispersion}
	2\cos(k L) = \Delta_L(\lambda),
\end{equation}
from which it is clear that \eqref{eq:dispersion1} admits a real-valued solution $\lambda(k)$ for real $k$ if and only if $\Delta_L(\lambda)\in[-2,2]$. The function $k\mapsto \lambda(k)$ of solutions is many valued and the branches are the \textit{spectral band functions} $\lambda_j(k)$ depicted, for example, in Fig.~\ref{fig:hero}(a).

The reappearance of the interval $[-2,2]$ in the condition $\Delta_L(\lambda)\in[-2,2]$ for solutions to the dispersion relation \eqref{eq:dispersion} is not a coincidence. Our objective now is to build a sequence of differential operators whose discriminants satisfy the recursion relation \eqref{eq:recursion}. Then, the appropriate formula for the dependence of $x_2$ on $x_0$ and $x_1$ will follow from the choice of $V$ in \eqref{eq:eigenproblem} (and the associated monodromy matrix).

A sequence of differential operators whose discriminants satisfy the recursion relation \eqref{eq:recursion} can be built using a \emph{Fibonacci tiling rule}, based on insight from \cite{kohmoto1983localization}. In this rule, we will define a sequence of periodic potentials $V_n(x)$, $n=0,1,2,\dots$ iteratively. Suppose that we have two initial potentials $V_0(x)$ and $V_1(x)$. Then, we can form a new periodic potential whose unit cells is the combination of the unit cells from  $V_0(x)$ and $V_1(x)$ side by side, as depicted in Fig.~\ref{fig:Fibonacci}. If we repeat the pattern of forming new periodic unit cells by combining the previous two, then we are applying a \emph{Fibonacci tiling rule}. This can, equivalently, be viewed as a two-letter substitution rule given by $A\to AB$ and $B\to A$ \cite{kolavr1993new, levine1984quasicrystals}. When initiated with the starting value of $\F_1=A$ (or, equivalently, $\F_0=B$), this gives the sequence of words shown in Fig.~\ref{fig:Fibonacci}. We use this to construct periodic potentials by assigning the value of $V(x)$ on each section of unit size according to the corresponding letter in the tiling rule. The length of each unit cell in the sequence follows the Fibonacci sequence. The link with \eqref{eq:recursion} was first discovered by \cite{kohmoto1983localization} and is straightforward to show from the observation that monodromy matrices for periodic systems combined in this way are related by $M_{F_n}=M_{F_{n-2}}M_{F_{n-1}}$. Rearranging this to give $M_{F_{n-2}}^{-1}=M_{F_{n-1}}M_{F_n}^{-1}$ and adding it to $M_{F_{n+1}}=M_{F_{n-1}}M_{F_{n}}$ gives 
\begin{equation}
	M_{F_{n+1}}+M_{F_{n-2}}^{-1}=M_{F_{n-1}}M_{F_{n}}+M_{F_{n-1}}M_{F_n}^{-1}.
\end{equation}
Taking the trace and using the fact that $\tr(M_{F_{n-1}}M_{F_{n}})+\tr(M_{F_{n-1}}M_{F_n}^{-1}) =\tr(M_{F_{n-1}})\tr(M_{F_{n}})$ gives \eqref{eq:recursion}, with $x_n=\tr(M_{F_n})$. Now that we have constructed a sequence of differential operators that satisfies the dynamical system of interest, we can use their densities of states to predict the statistics. 

\section{Canonical example} \label{sec:constants}

Before computing density of states, we pause to examine some specific examples of dependencies (of $x_2$ on $x_0$ and $x_1$) that allow the system to be related to a sequence of periodic operators. The formula \eqref{eq:x2} arises by considering the canonical example of constant potentials $V_0(x)=A>0$ and $V_1(x)=B>0$. These are, trivially, 1-periodic and their monodromy matrices are given by
\begin{equation}
	M_{F_j}=\begin{pmatrix} \cos\omega_j & \frac{1}{\omega_j}\sin\omega_j \\ -\omega_j\sin\omega_j & \cos\omega_j \end{pmatrix}, \quad \text{where } j=0,1,
\end{equation}
with $\omega_0=\sqrt{B-\lambda}$ and $\omega_1=\sqrt{A-\lambda}$. Here $F_0=1$ and $F_1=1$ are taken as the first two (nonzero) Fibonacci numbers. The associated discriminants are given by
\begin{equation}
	\Delta_{F_0}=2\cos\omega_0
	\quad\text{and}\quad
	\Delta_{F_1}=2\cos\omega_1.
\end{equation}
Suppose we form a new material which is 2-periodic and alternates between the values $A$ and $B$ on each unit-sized section (as sketched in Fig.~\ref{fig:Fibonacci}). We can compute the monodromy matrix of this 2-periodic system through the product $M_{F_2}=M_{F_0}M_{F_1}$. Taking the trace of $M_{F_2}$ gives the discriminant
\begin{equation}
	\Delta_{F_2} = 2\cos\omega_0\cos\omega_1-\big( \tfrac{\omega_0}{\omega_1} + \tfrac{\omega_1}{\omega_0} \big) \sin\omega_0\sin\omega_1.
\end{equation}
This is the formula that was specified in \eqref{eq:x2} to fix the dependence of $x_2$ on $x_1$ and $x_0$. Thus, by choosing this relationship between the initial values, it is guaranteed that $x_2=\Delta_{F_2}(\lambda)$.

There are many other formulas that would similarly allow for the orbits of \eqref{eq:recursion} to be related to a sequence of periodic operators. Indeed, any choice for $V_0$ and $V_1$ would yield such a formula (although it may not admit a concise analytic expression). An example that uses sinusoidal potentials to yield a formula in terms of Mathieu functions is derived in Appendix~\ref{app:mathieu}.

\begin{figure}
	\includegraphics[width=\linewidth]{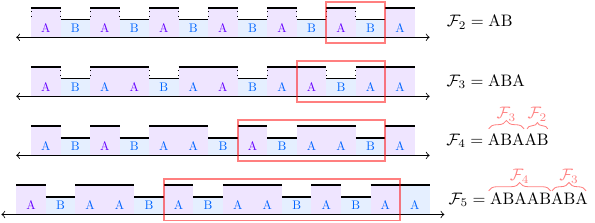}
	\caption{The Fibonacci tiling rule. Each new periodic potential is formed by combining the previous two unit cells in the sequence.} \label{fig:Fibonacci}
\end{figure}

\section{Density of states} \label{sec:density}

The density of states for a periodic differential operator is well studied. The density $D(\lambda)$ quantifies the number of states per unit energy. For an infinitely periodic system it is typically defined in terms of the limit of finite systems with arbitrarily large size \cite{kuchment2016overview}. At a given value $\lambda$, the density of states is the reciprocal of the slope of the spectral band function at that point:
\begin{equation} \label{eq:DoS}
	D(\lambda)=\frac{1}{2\pi} \left|\dd{}{k}\lambda_j|_{\lambda_j=\lambda}\right|^{-1},
\end{equation}
That there can only be at most one point on the dispersion curves where $\lambda_j=\lambda$ is a property of one-dimensional systems \cite{reed1978iv}.

From \eqref{eq:DoS}, it is apparent that singularities will occur whenever a band function $k\mapsto\lambda_j(k)$ has zero slope. These are known as \textit{van Hove singularities} \cite{van1953occurrence}. In general, for one-dimensional operators this will happen at the edges of the reduced Brillouin zone. This can be seen by differentiating the dispersion relation (\ref{eq:dispersion}) with respect to $k$ to give
\begin{equation}
	-2L\sin(kL)=\dd{\Delta_L}{\lambda}\dd{\lambda}{k},
\end{equation}
from which it is clear that $\mathrm{d}\lambda/\mathrm{d}k=0$ if and only if $k\in\{0,\pm\pi/L\}$ (whenever $\mathrm{d}\Delta_L/\mathrm{d}\lambda$ is non-zero).

Since the values of the discriminants associated to eigenstates of the periodic operators can be related to the orbits of the recursion relation \eqref{eq:recursion}, we would like to use similar tools to characterise its statistics. We can use our understanding of the map $\lambda\mapsto\Delta_L$ to identify the value of the discriminant at each point on one of the spectral band functions $\lambda_j$. In light of this, the integrated density of states (IDS) $\nu(\Delta)$ associated to the value of the discriminant can be written as 
\begin{equation} \label{eq:proof1}
	\nu(\Delta) = \sum_{j=1}^\infty \mu\{k\in\mathcal{B}:\Delta_L(\lambda_j(k))<\Delta\},
\end{equation}
where $\mu$ is the normalised Haar measure $(2\pi)^{-1}\mathrm{d} k$. Since the band functions are monotonic, $\Delta_L$ is a well-defined function of $k\in\mathcal{B}$ and the dependence on $\lambda_j$ in \eqref{eq:proof1} can be dropped. Hence, (\ref{eq:proof1}) can be simplified to
\begin{equation} \label{eq:proof2}
	\nu(\Delta) = \mu\{k\in\mathcal{B}:\Delta_L(k)<\Delta\},
\end{equation}
for $-2\leq\Delta\leq2$.

The density of states can be calculated as the derivative of the IDS \eqref{eq:proof2}. This gives
\begin{equation} \label{eq:proof3}
	D(\Delta)=\dd{\nu}{\Delta}=\frac{1}{2\pi} \left|\dd{\Delta_L}{k}\right|^{-1}.
\end{equation}
Differentiating the dispersion relation $\Delta_L=2\cos kL$ with respect to $k$ gives
\begin{equation*}
	\dd{\Delta_L}{k}=-2L\sin kL=-L\sqrt{4-4\cos^2kL}=-L\sqrt{4-\Delta_L^2}.
\end{equation*}
Substituting this into \eqref{eq:proof3} and normalising to give a density function on the interval $(-2,2)$ gives a formula for the density of states:
\begin{equation} \label{eq:main}
	D(x)=\frac{1}{\pi}\frac{1}{\sqrt{4-x^2}},
\end{equation}
for $-2<x<2$. This density function is plotted (dashed line) in Fig.~\ref{fig:histograms}(d) where we see excellent agreement with the limiting statistics of the chaotic dynamical system.

\section{Conclusions}

This work showed that constructing a sequence of periodic differential operators associated to a nonlinear recursion relation provides an avenue for deriving an explicit formula for the limiting statistics. For the example studied here, this formula was shown to predict the statistics from numerical experiments accurately. This result also offers an explanation for the clustering of states of the chaotic recursion relation near to the identifiable critical points at $x=\pm2$, as it illuminates an equivalence with Van Hove singularities in the densities of states of the associated operators.

We envisage that this strategy could be extended to other systems in future work. It was recently shown that a sequence of periodic operators can be constructed in association with the chaotic logistic map \cite{davies2024spectral}. As a result, the familiar singularities in the well-known invariant measure of that system \cite{jakobson1981absolutely, sethna2021statistical} are also equivalent to van Hove singularities. Other straightforward generalisations could also be generated easily. For example, the method could be generalised easily to the nonlinear recursion relations associated to the \emph{generalised} Fibonacci tiling rules\footnote{See \cite{kola1990one} for details of the generalised Fibonacci tiling rules, \cite{gumbs1988dynamical, gumbs1988scaling, gumbs1989electronic} for the basic spectral properties of the associated operators and \cite{davies2023super} for analysis of the escape criteria.}. However, generalising this work to a wider set of nonlinear recursion relation is a challenging open question for future study.


\appendix
\section{Escape criterion} \label{app:escape}

We present the following modification of the result proved in \cite[Theorem 3.1]{davies2023super}, to show that the growth condition \eqref{eq:growth} is an escape criterion for the system.

\begin{lemma} 
	Suppose that $\{x_n:n\in\Znn\}$ satisfies the recursion relation (\ref{eq:recursion}). If there exists some $N\in\Znn$ such that 
	$$ 2<|x_N|, \quad |x_N|<|x_{N+1}|, \quad |x_{N+1}|<|x_{N+2}|,$$
	then $|x_n|>2$ for all $n\geq N$ and $|x_n|\to\infty$ as $n\to\infty$.
\end{lemma}
\begin{proof}
	It is easy to see that 
	\begin{align*}
	|x_{N+3}|&\geq |x_{N+2}||x_{N+1}|-|x_N| \\
	&\geq |x_{N+2}||x_{N}|-|x_N| \\
	& = |x_{N+2}|(|x_{N}|-1)+(|x_{N+2}|-|x_N|).
	\end{align*}
	By hypothesis, $|x_{N+2}|>|x_N|$ meaning that $|x_{N+3}|>|x_{N+2}|(|x_{N}|-1),$
	where $(|x_{N}|-1)>1$. Proceeding by induction, it can be shown that
	\begin{equation}
		|x_{N+2+n}|>|x_{N+2}|(|x_{N}|-1)^n,
	\end{equation}
	for all $n\in\Zp$, from which it is clear that $|x_n|\to\infty$ as $n\to\infty$.
\end{proof}

\section{Derivation of $f_{X_3}$} \label{app:pdf_indep}

In this appendix, the probability density function $f_{X_3}$ given in \eqref{eq:indepdist} is calculated, where $X_3=X_2X_1-X_0$ with $X_0$, $X_1$ and $X_2$ being independent and identically distributed random variables, drawn from uniform distributions on the interval $[-2,2]$.

First, the density function of the product $X_1X_2$ can be calculated using elementary properties of conditional probability. In particular, it holds that
\begin{equation} \label{eq:B1}
	f_{X_2X_1}(x)=\int_{-\infty}^\infty f_{X_1}(y) f_{X_2}(x/y) \frac{1}{|y|} \de y.
\end{equation}
Here, $f_{X_1}(y)=1/4$ for $|y|\leq2$ and is zero otherwise. Similarly, $f_{X_2}(x/y)=1/4$ for $|x/y|<2$. Hence, the integrand in \eqref{eq:B1} is non-zero when $|x|/2\leq |y| \leq 2$. Using symmetry of the integrand, this leads to
\begin{align}
	f_{X_2X_1}(x)&=2\int_{|x|/2}^2 \frac{1}{4}\cdot \frac{1}{4}\cdot \frac{1}{y} \de y \nonumber \\
	&=\frac{1}{8}\left(\log2-\log\frac{|x|}{2}\right) \nonumber \\
	&=-\frac{1}{8}\log\left(\frac{|x|}{4}\right), \label{eq:pdf_X1X2}
\end{align}
for $-4\leq x\leq 4$ and vanishes outside of this interval.

Since $X_2X_1$ and $X_0$ are independent, the density function of their difference $X_3=X_2X_1-X_0$ is given by the convolution of $f_{X_2X_1}$ and $f_{X_0}$:
\begin{align*}
	f_{X_3}(x)&=\int_{-\infty}^{\infty} f_{X_2X_1}(x-z) f_{X_0}(z) \de z\\ &= -\frac{1}{32}\int_{-2}^{2} \log\left(\frac{|x-z|}{4}\right)\chi_{(-4,4)}(x-z) \de z.
\end{align*}
If $x,z\in[-2,2]$, then $|x-z|\leq 4$, so $\chi_{(-4,4)}(x-z)=1$. Thus, for $-2\leq x\leq 2$ it holds that
\begin{equation}
	f_{X_3}(x)=-\frac{1}{32}\int_{-2}^{2} \log\left(\frac{|x-z|}{4}\right) \de z.
\end{equation}
Partitioning the integral into $[-2,x]$ and $[x,2]$ gives
\begin{equation*}
	f_{X_3}(x)=-\frac{1}{32}\int_{-2}^{x} \log\left(\tfrac{x-z}{4}\right) \de z-\frac{1}{32}\int_{x}^{2} \log\left(\tfrac{z-x}{4}\right) \de z,
\end{equation*}
where each of the two integrals can be solved using integration by parts to obtain the formula stated in \eqref{eq:indepdist}:
\begin{equation*}
	f_{X_3}(x)= \frac{x-2}{32}\log\left(\frac{2-x}{4}\right) - \frac{x+2}{32}\log\left(\frac{x+2}{4}\right) + \frac{1}{8},
\end{equation*}
for $-2\leq x\leq2$.

\section{Mathieu functions} \label{app:mathieu}

Another reasonable choice of potential function is piecewise sinusoidal on each of the two unit cells to be tiled. Consider, for example, the $\pi$-periodic potential
\begin{equation}
	V_q(x)=\alpha-2q\cos(2x),
\end{equation}
for $\alpha, q\in\mathbb{R}$. Then the monodromy matrix can be written in terms of the Mathieu functions $C(a,q,x)$ and $S(a,q,x)$, which are defined as the even and odd solutions to the Mathieu equation $u''(x)+(a-2q\cos(2x))u(x)=0$ (see e.g.
\cite{abramowitz1948handbook} for their properties). By the symmetry of the equation, we have that $C'(a,q,0)=0$ and $S(a,q,0)=0$, so we can choose the normalisation to be such that they are exactly the basis functions: $\phi=C$ and $\psi=S$. Then, the monodromy matrix is given by
\begin{equation} \label{eq:monod_mathieu}
	M_\pi^{(q)}(\lambda)=
	\begin{pmatrix} C(\alpha-\lambda,q,\pi) & S(\alpha-\lambda,q,\pi) \\ C'(\alpha-\lambda,q,\pi) & S'(\alpha-\lambda,q,\pi) \end{pmatrix} \\
\end{equation}
where $C'$ and $S'$ denote the derivatives with respect to the position variable $x$. 

Given two initial values $x_0,x_1\in[-2,2]$, we need to find $\lambda$, $\alpha$, $q_0$ and $q_1$ such that
\begin{equation*}
	x_n=\tr\left( M_\pi^{(q_n)}(\lambda) \right)
	=C(\alpha{-}\lambda,q_n,\pi)+S'(\alpha{-}\lambda,q_n,\pi).
\end{equation*}
This is always possible, not least since the range of this function is larger than $[-2,2]$. Then, the required formula for $x_2$ is
\begin{align}
	&x_2=\tr\left(  M_\pi^{(q_1)}(\lambda) M_\pi^{(q_0)}(\lambda) \right) \\
	&= C(q_1)C(q_0){+}S(q_1)C'(q_0){+}C'(q_1)S(q_0)  {+}S'(q_1)S'(q_0), \nonumber
\end{align}
where we have written $C(q_n):=C(\alpha-\lambda,q_n,\pi)$, and similarly for $S$, for brevity.

We can apply the Fibonacci tiling rule to two unit cells corresponding to $q=q_0$ and $q=q_1$. Then, the successive transfer matrices will satisfy the properties explored above, meaning that their traces satisfy the recursion relation \eqref{eq:recursion} and the density of states converges to the distribution \eqref{eq:main}.

\section{Distribution of $x_2$} \label{app:pdf_x2}

This appendix studies the distribution of the random variable $X_2$, given by
\begin{equation} \label{eq:X3}
	X_2 = 2\cos\Omega_0\cos\Omega_1-\left( \frac{\Omega_0}{\Omega_1} + \frac{\Omega_1}{\Omega_0} \right) \sin\Omega_0\sin\Omega_1,
\end{equation}
where $\Omega_0$ and $\Omega_1$ are given by
\begin{equation}
	\Omega_0=\arccos(X_0/2), \quad \Omega_1=\arccos(X_1/2),
\end{equation}
with $X_0$ and $X_1$ being independent and identically distributed random variables, obeying uniform distributions on the interval $[-2,2]$. It will not be possible to derive an explicit formula for the distribution of $X_2$, but the features of its distribution (shown in Fig.~\ref{fig:apphists}(e)) can be better understood by examining each term in (\ref{eq:X3}).

The first term of (\ref{eq:X3}) can be understood easily, since $2\cos\Omega_0\cos\Omega_1=\frac{1}{2}X_0X_1$. Thus, it is possible to rescale the density function (\ref{eq:pdf_X1X2}) to see that $2\cos\Omega_0\cos\Omega_1$ has probability density function $f_{2\cos\Omega_0\cos\Omega_1}$ given by
\begin{equation} \label{eq:coscosdist}
	f_{2\cos\Omega_0\cos\Omega_1}(x)=-\frac{1}{4}\log\left(\frac{|x|}{2}\right),
\end{equation}
for $-2\leq x\leq 2$, vanishing outside of this interval. This distribution is shown in Fig.~\ref{fig:apphists}(a), where the histogram shows $10^6$ independent realisations of $2\cos\Omega_0\cos\Omega_1$ and the dotted line shows the density function (\ref{eq:coscosdist}).

\begin{figure*}
	\begin{subfigure}[b]{0.3\textwidth}
		\includegraphics[width=0.95\linewidth]{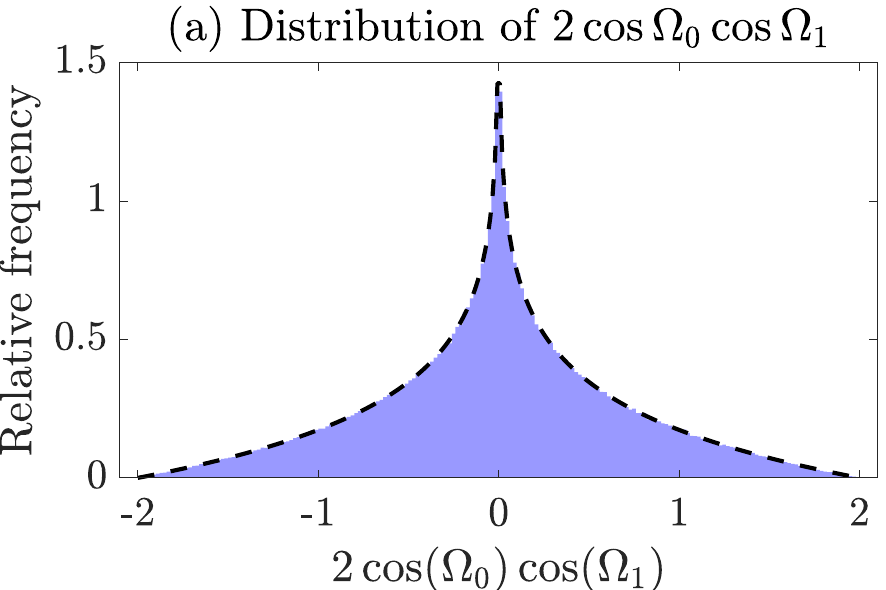}
	\end{subfigure}
	\begin{subfigure}[b]{0.3\textwidth}
		\includegraphics[width=0.95\linewidth]{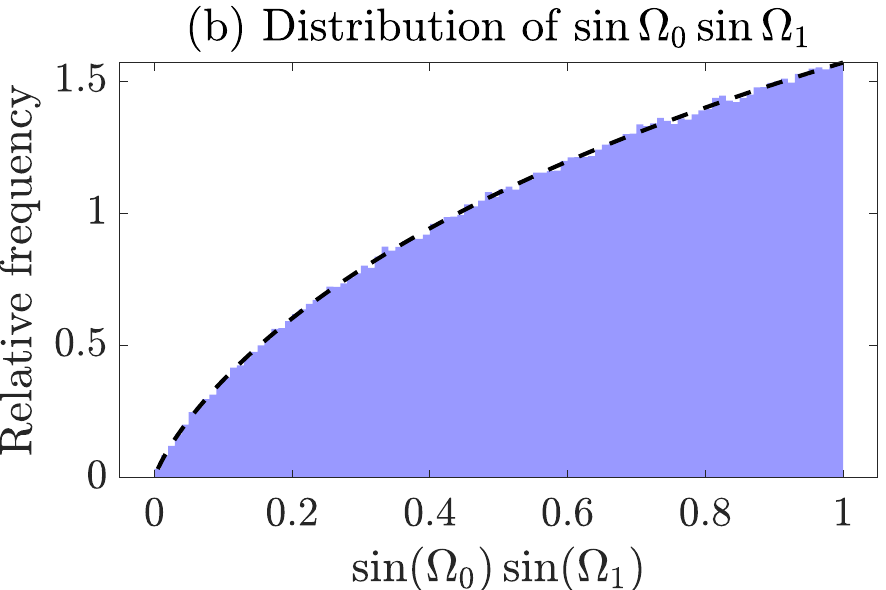}
	\end{subfigure}
	\begin{subfigure}[b]{0.3\textwidth}
		\includegraphics[width=0.95\linewidth]{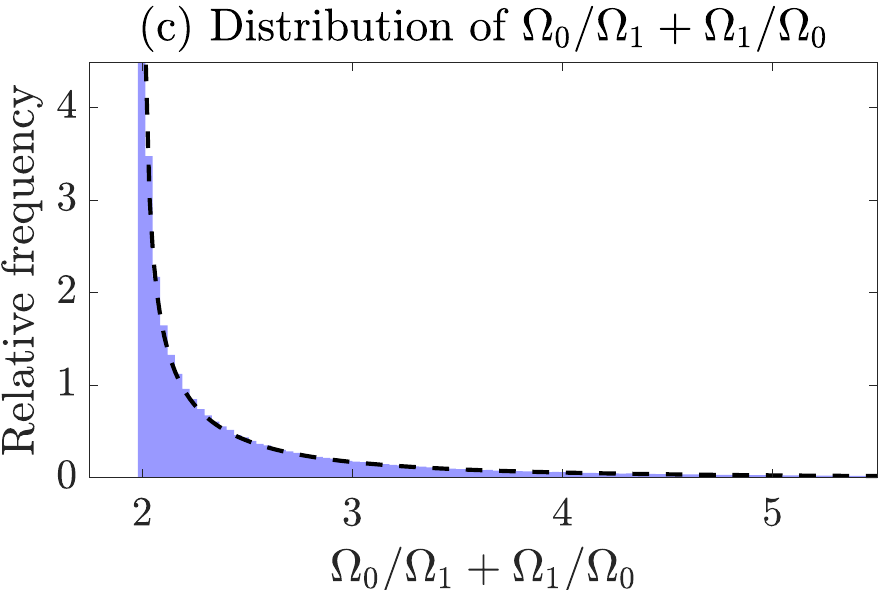}
	\end{subfigure}
	
	\vspace{0.4cm}

	\begin{subfigure}[b]{0.3\textwidth}
		\includegraphics[width=0.95\linewidth]{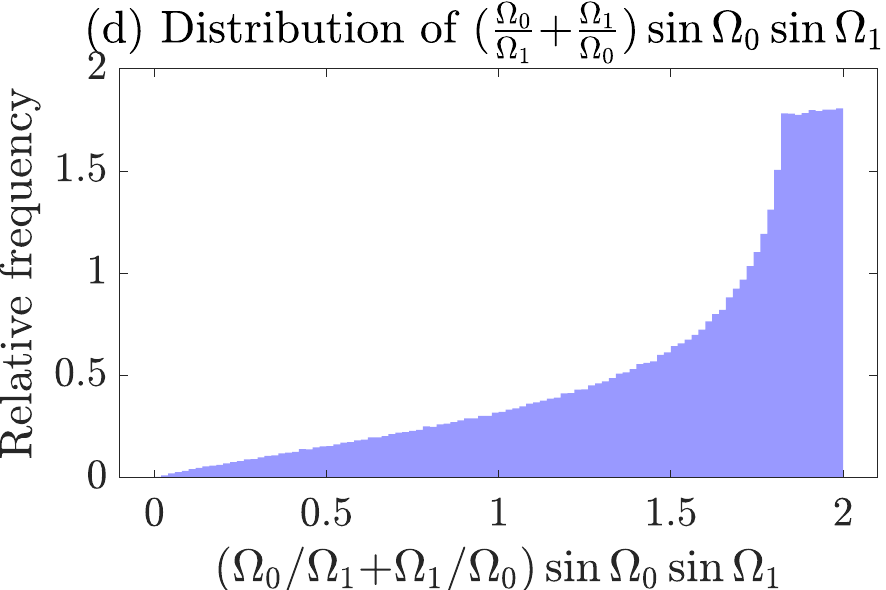}
	\end{subfigure}
	\begin{subfigure}[b]{0.6\textwidth}
		\includegraphics[width=0.95\linewidth]{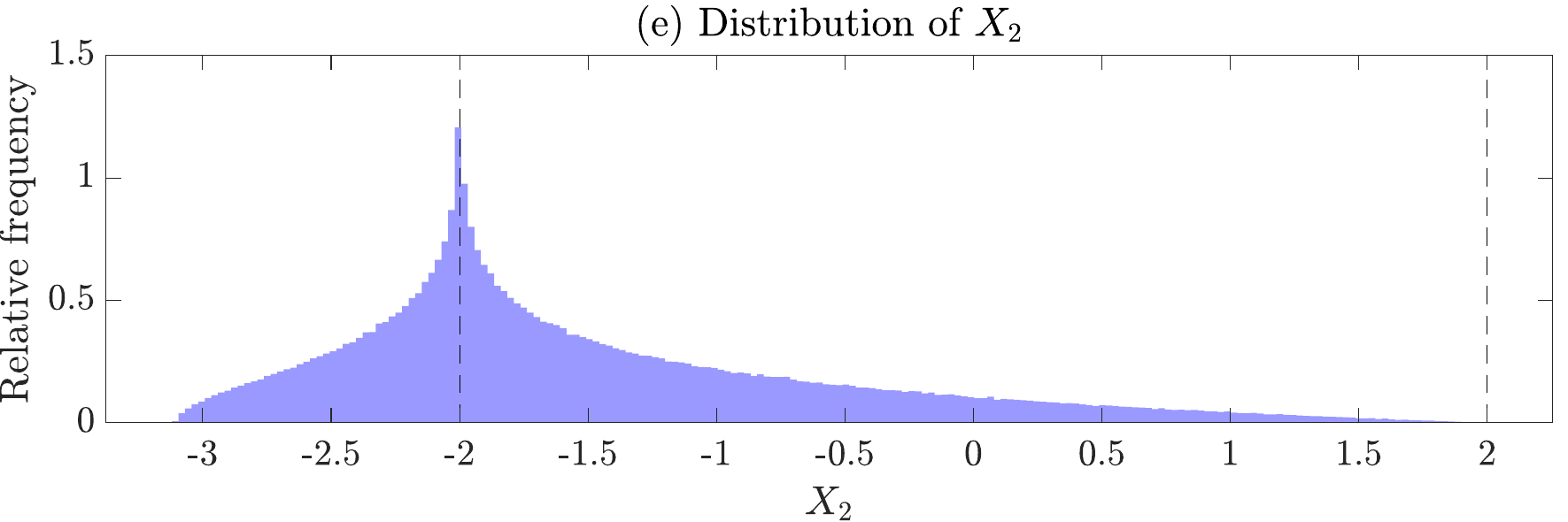}
	\end{subfigure}
	
	\vspace{0.4cm}
	
	\begin{subfigure}[b]{0.45\textwidth}
		\includegraphics[width=0.95\linewidth]{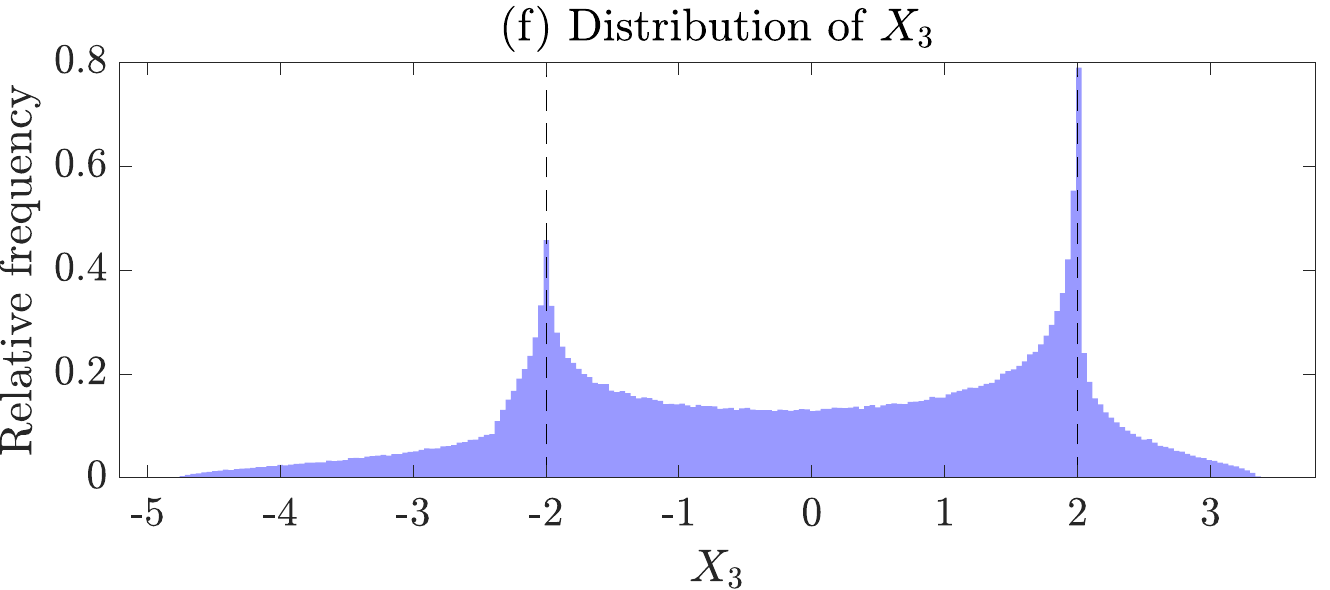}
	\end{subfigure}
	\begin{subfigure}[b]{0.45\textwidth}
		\includegraphics[width=0.95\linewidth]{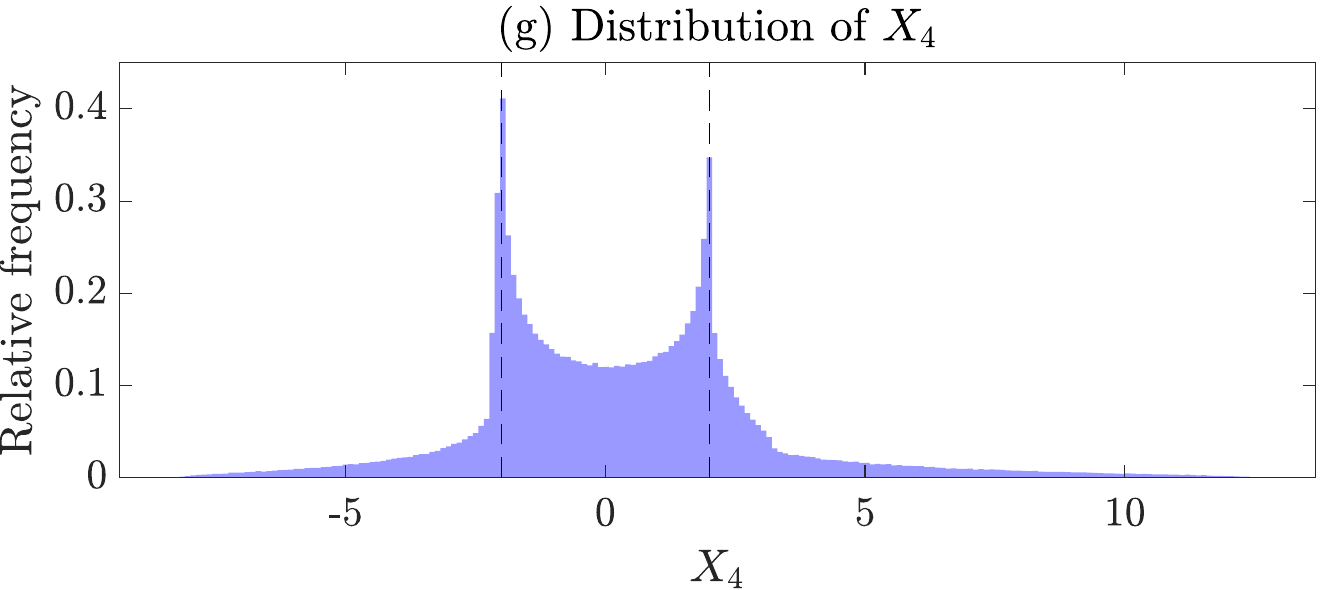}
	\end{subfigure}
	\caption{The distributions of the first few terms of the recursion relation (\ref{eq:recursion}) (and their key components). In (a), (b) and (c) the analytically computed distributions from (\ref{eq:coscosdist}), (\ref{eq:sinsindist_result}) and (\ref{eq:1221dist}), respectively, are shown in dashed lines. In (e), (f) and (g) the critical thresholds $\pm2$ are marked by vertical dashed lines. In each case, $10^5$ independent realisations of the initial conditions are used.} \label{fig:apphists}
\end{figure*}

To understand the second term of (\ref{eq:X3}), the first step is to compute the distribution of $\sin\Omega_0\sin\Omega_1$. If $X_i$ is uniformly distributed on $[-2,2]$, then $\Omega_i=\arccos(X_i/2)$ has the probability density function
\begin{equation} \label{eq:omegadist}
	f_{\Omega_i}(\omega) = \frac{1}{2} \sin\omega,
\end{equation} 
for $0\leq\omega\leq\pi$ and vanishes outside of this interval. Subsequently, the random variable $\sin(\Omega_i)$ has the probability density function
\begin{equation} \label{eq:sinomegadist}
	f_{\sin(\Omega_i)}(y) = \tan(\arcsin y) = \frac{y}{\sqrt{1-y^2}},
\end{equation}
for $0\leq y\leq1$, which vanishes outside of this interval. The product $\sin\Omega_0\sin\Omega_1$ can be handled by computing the cumulative distribution function
\begin{equation*} 
	P(\sin\Omega_0\sin\Omega_1\leq x) = \int_0^1 P(y \sin\Omega_0\leq x) \frac{y}{\sqrt{1-y^2}}\de y.
\end{equation*}
If $x/y<1$ then $P(\sin\Omega_0\leq x/y)=1-\sqrt{1-x^2/y^2}$. Conversely, if $x/y>1$ then $P(\sin\Omega_0\leq x/y)=1$. Hence, the integral above can be partitioned into the intervals $[0,x]$ and $[x,1]$ to see that 
\begin{equation} 
	P(\sin\Omega_0\sin\Omega_1\leq x) = 1-\int_x^1 \frac{\sqrt{y^2-x^2}}{\sqrt{1-y^2}} \de y.
\end{equation}
To compute the density function, this expression should be differentiated with respect to $x$, for which the Leibniz integral rule can be used to see that 
\begin{equation} \label{eq:sinsindist2}
	f_{\sin\Omega_0\sin\Omega_1}(x) = \int_x^1 \frac{x}{\sqrt{y^2-x^2}\sqrt{1-y^2}} \de y.
\end{equation}
This integral can be expressed in terms of the complete elliptic integral of the first kind, which is the function $K:(0,1)\to\mathbb{R}$ defined as \cite{abramowitz1948handbook}
\begin{equation}
	K(k)=\int_0^{\frac{\pi}{2}} \frac{1}{\sqrt{1-k^2\sin^2 t}} \de t.
\end{equation}
Making an appropriate substitution and change of variables in the integral (\ref{eq:sinsindist2}) gives that
\begin{equation} \label{eq:sinsindist_result}
	f_{\sin\Omega_0\sin\Omega_1}(x) = xK(1-x^2).
\end{equation}
This distribution is shown in Fig.~\ref{fig:apphists}(b), where the histogram shows $10^6$ independent realisations of $\sin\Omega_0\sin\Omega_1$ and the dotted line is the the density function (\ref{eq:sinsindist_result}).

The missing piece needed to understand the distribution of $X_2$ is the density function of $\Omega_0/\Omega_1 + \Omega_1/\Omega_0$. Given two independent random variables with density given by (\ref{eq:omegadist}), their ratio has distribution
\begin{equation} \label{eq:ratiodist}
	P\left(\frac{\Omega_0}{\Omega_1}\leq x\right) = \int_0^\pi P(\Omega_0\leq xy) \frac{1}{2}\sin y \de y,
\end{equation}
where $x>0$. If $0<\pi/x<y$ then $P(\Omega_0\leq xy)=1$. Conversely if $0<y<\pi/x$ then 
\begin{equation} \label{eq:distcalc}
	P(\Omega_0\leq xy) = \int_0^{xy} \frac{1}{2} \sin t \de t = \frac{1}{2}-\frac{1}{2}\cos xy.
\end{equation}
If $x<1$ then it must be the case that $y<\pi<\pi/x$ so it can be calculated directly that
\begin{align} \label{eq:ratioCDF1}
	P\left(\frac{\Omega_0}{\Omega_1}\leq x\right) &= \int_0^\pi \left( \frac{1}{2}-\frac{1}{2}\cos xy \right) \frac{1}{2}\sin y \de y  \nonumber \\&=\frac{1}{2}+\frac{1+\cos\pi x}{4(x^2-1)},
\end{align}
which can be differentiated to find that $\Omega_0/\Omega_1$ has probability density function given by
\begin{equation*}
	f_{\Omega_0/\Omega_1}(x)=-\frac{\pi(x^2-1)\sin\pi x+2x(1+\cos\pi x)}{4(x^2-1)^2},
\end{equation*}
for $0<x<1$. For $x>1$, the integral in (\ref{eq:ratiodist}) can be handled by partitioning it into the intervals $[0,\pi/x]$ and $[\pi/x,\pi]$:
\begin{eqnarray} \label{eq:ratioCDF2}
	P\left(\frac{\Omega_0}{\Omega_1}\leq x\right) &= \int_0^{\pi/x} \left( \frac{1}{2}-\frac{1}{2}\cos xy \right)\frac{1}{2}\sin y \de y \nonumber \\
	&+ \int_{\pi/x}^\pi \frac{1}{2}\sin y \de y \nonumber \\
	& = \frac{3+\cos\frac{\pi}{x}}{4} + \frac{1+\cos\frac{\pi}{x}}{4(x^2-1)},
\end{eqnarray}
which can be differentiated to obtain that
\begin{equation}
	f_{\Omega_0/\Omega_1}(x)=\frac{\pi(x^2-1)\sin\frac{\pi}{x}-2x(1+\cos\frac{\pi}{x})}{4(x^2-1)^2},
\end{equation}
for $x>1$. Subsequently, the distribution of $\Omega_0/\Omega_1 + \Omega_1/\Omega_0$ can be found by applying the transformation $x\mapsto x+1/x$. This has a global minimum for $x>0$ at $x=1$ so it holds that
\begin{equation}
	P\left(\frac{\Omega_0}{\Omega_1}+\frac{\Omega_1}{\Omega_0}\leq 2\right)=0.
\end{equation}
%
%
\begin{widetext}
Meanwhile, for any $x\geq2$, it holds that
\begin{align}
	P\left(\frac{\Omega_0}{\Omega_1}+\frac{\Omega_1}{\Omega_0}\leq x\right)
	&=P\left(\frac{x-\sqrt{x^2-4}}{2}\leq\frac{\Omega_0}{\Omega_1}\leq \frac{x+\sqrt{x^2-4}}{2}\right) \nonumber \\
	&=P\left(\frac{\Omega_0}{\Omega_1}\leq \frac{x+\sqrt{x^2-4}}{2}\right) - P\left(\frac{\Omega_0}{\Omega_1}\leq\frac{x-\sqrt{x^2-4}}{2}\right).
\end{align}
Using (\ref{eq:ratioCDF1}) and (\ref{eq:ratioCDF2}), this becomes
\begin{equation}
	\begin{split}
		P\left(\frac{\Omega_0}{\Omega_1}+\frac{\Omega_1}{\Omega_0}\leq x\right) &= \frac{1}{4} + \frac{1}{(x+\sqrt{x^2-4})^2-4}\left(1+ \cos\left( \frac{2\pi}{x+\sqrt{x^2-4}} \right)\right)
		\\ &\quad +\frac{1}{4}\cos\left( \frac{2\pi}{x+\sqrt{x^2-4}} \right) - \frac{1}{(x-\sqrt{x^2-4})^2-4}\left(1+ \cos\left( \frac{\pi}{2}(x-\sqrt{x^2-4}) \right)\right)
	\end{split}
\end{equation}
which can be differentiated to obtain the probability density function
\begin{equation} \label{eq:1221dist}
	\begin{split}
		f_{\Omega_0/\Omega_1+\Omega_1/\Omega_0}(x)&=\frac{1}{4(x^2-4)^\frac{3}{2}} \bigg[ -4-2\cos\left( \frac{\pi}{2}(x-\sqrt{x^2-4})\right) \\
		&\quad -2\cos\left( \frac{2\pi}{x+\sqrt{x^2-4}}\right) +\pi \sqrt{x^2-4}\sin \left( \frac{\pi}{2}(x-\sqrt{x^2-4})\right) \\
		&\qquad +\pi \sqrt{x^2-4}\sin \left(  \frac{2\pi}{x+\sqrt{x^2-4}} \right) 
		\bigg].
	\end{split}
\end{equation}
\end{widetext}
%
%
The distribution \eqref{eq:1221dist} is shown with the dashed line in Fig.~\ref{fig:apphists}(c), where it is compared to a histogram showing $10^6$ independent realisations of $\Omega_0/\Omega_1+\Omega_1/\Omega_0$. 

The distribution of $X_2$ could now be obtained by a similar approach, however it is not reasonable to expect that explicit analytic expressions will be possible. Nevertheless, the density functions calculated above are already sufficient to yield some intuition on the qualitative properties of the distribution of $X_2$, which is shown in Fig.~\ref{fig:apphists}(e).

First, consider the distribution of $(\Omega_0/\Omega_1+\Omega_1/\Omega_0)\sin\Omega_0\sin\Omega_1$. This variable must fall in the interval $[0,2]$ since it holds that $0\leq (x/y+y/x)\sin x\sin y\leq 2$ for $0\leq x,y\leq\pi$. Further, this function attains its maximal value of 2 when $x=y=\pi/2$. Both $\Omega_0$ and $\Omega_1$ are drawn from distributions with sinusoidal density functions (\ref{eq:omegadist}) which attain maxima at $\pi/2$. This explains why values tend to cluster around the upper end of the interval $[0,2]$, as shown in Fig.~\ref{fig:apphists}(d). The sharpness of this peak can be understood by examining the sharp peaks in the densities of $\sin\Omega_i$ (\ref{eq:sinomegadist}) and $\Omega_0/\Omega_1+\Omega_1/\Omega_0$ (\ref{eq:1221dist}): the values of $\sin\Omega_0$ and $\sin\Omega_1$ are strongly clustered around 1 and the values of $\Omega_0/\Omega_1+\Omega_1/\Omega_0$ are strongly clustered around 2, leading to a peak in the density of their product close to 2.

The distribution of $X_2$, shown in Fig.~\ref{fig:apphists}(e), is the difference between the distributions of $2\cos\Omega_0\cos\Omega_1$ in Fig.~\ref{fig:apphists}(a) and of $(\Omega_0/\Omega_1+\Omega_1/\Omega_0)\sin\Omega_0\sin\Omega_1$ in Fig.~\ref{fig:apphists}(d). In this case, there is strong clustering of values around -2. This can be understood from the fact that the density function of $2\cos\Omega_0\cos\Omega_1$ in (\ref{eq:coscosdist}) has a singularity at 0 and the distribution of $(\Omega_0/\Omega_1+\Omega_1/\Omega_0)\sin\Omega_0\sin\Omega_1$ in (\ref{eq:1221dist}) shows clustering around 2. This clustering of values of $X_2$ around -2 is the first hint of the van Hove singularities studied in this work emerging. Examining the distributions of $X_3$ and $X_4$ in Fig.~\ref{fig:apphists}(f) and Fig.~\ref{fig:apphists}(g), the singularities around $\pm2$ are increasingly clear. As shown in Fig.~\ref{fig:histograms}, as $n$ becomes large the distribution of $X_n$ confined to the interval $[-2,2]$ loses its asymmetry and the clustering of values around $\pm2$ converges to singularities in the density of states.

\vspace{1em}

\begin{acknowledgments}
This work was partly supported by the Engineering and Physical Sciences Research Council (EPSRC) through a Research Fellowship with grant number EP/X027422/1. The author would like to thank Ian Hooper, whose comments were the inspiration for this work, as well as Matteo Tanzi for helpful discussions. The codes used to obtain the numerical results presented here are available for download \cite{mycode}. 
\end{acknowledgments}

\bibliography{references}

\end{document}